\documentclass[12pt,reqno]{amsart}
\usepackage{graphicx}
\usepackage{amscd}
\usepackage{amsmath}
\usepackage{epsfig}
\usepackage{amsfonts}
\usepackage{amssymb}

\providecommand{\U}[1]{\protect\rule{.1in}{.1in}}
\providecommand{\U}[1]{\protect\rule{.1in}{.1in}}
\allowdisplaybreaks
\newtheorem{theorem}{Theorem}
\theoremstyle{plain}

\newtheorem{lemma}{Lemma}

\numberwithin{equation}{section}

\input{tcilatex}

\begin{document}
\title[Dynamical Invariants]{Dynamical Invariants for Variable Quadratic
Hamiltonians}
\author{Sergei K. Suslov}
\address{School of Mathematical and Statistical Sciences \& Mathematical,
Computational and Modeling Sciences Center, Arizona State University, Tempe,
AZ 85287--1804, U.S.A.}
\email{sks@asu.edu}
\urladdr{http://hahn.la.asu.edu/\symbol{126}suslov/index.html}
\date{\today }

\begin{abstract}
We consider linear and quadratic integrals of motion for general variable
quadratic Hamiltonians. Fundamental relations between the eigenvalue problem
for linear dynamical invariants and solutions of the corresponding Cauchy
initial value problem for the time-dependent Schr\"{o}dinger equation are
emphasized. An eigenfunction expansion of the solution of the initial value
problem is also found. A nonlinear superposition principle for generalized
Ermakov systems is established as a result of decomposition of the general
quadratic invariant in terms of the linear ones.
\end{abstract}

\maketitle

\section{An Introduction}

Quantum systems with variable quadratic Hamiltonians are called the
generalized harmonic oscillators (see, for example, \cite{Berry85}, \cite%
{Dod:Mal:Man75}, \cite{Dodonov:Man'koFIAN87}, \cite{Hannay85}, \cite{Leach90}%
, \cite{Malkin:Man'ko79}, \cite{Wolf81}, \cite{Yeon:Lee:Um:George:Pandey93}
and references therein). They have attracted substantial attention over the
years in view of their great importance in many advanced quantum problems.
Examples include coherent states and uncertainty relations \cite%
{Malkin:Man'ko79}, \cite{Malk:Man:Trif69}, \cite{Malk:Man:Trif70}, \cite%
{Mand:Karpov:Cerf}, \cite{Klauder:Sudarshan}, Berry's phase \cite{Berry85}, 
\cite{Berry:Hannay88}, \cite{Hannay85}, \cite{Leach90}, \cite{Morales88},
asymptotic and numerical methods \cite{Kruskal62}, \cite{Maslov:Fedoriuk}, 
\cite{Milne30}, \cite{Mun:Ru-Paz:Wolf09}, quantization of mechanical systems 
\cite{Faddeyev69}, \cite{FeynmanPhD}, \cite{Feynman}, \cite{Fey:Hib}, \cite%
{Kochan07}, \cite{Kochan10} and Hamiltonian cosmology \cite%
{Bertoni:Finelli:Venturi98}, \cite{Finelli:Gruppuso:Venturi99}, \cite%
{Finelli:Vacca:Venturi98}, \cite{Hawkins:Lidsey02}, charged particle traps 
\cite{Major:Gheorghe:Werth} and motion in uniform magnetic fields \cite%
{Cor-Sot:Lop:Sua:Sus}, \cite{Corant:Snyder58}, \cite{Dodonov:Man'koFIAN87}, 
\cite{La:Lif}, \cite{Lewis67}, \cite{Lewis68}, \cite{Lewis:Riesen69}, \cite%
{Malk:Man:Trif70}, molecular spectroscopy \cite{Dokt:Mal:Man77}, \cite%
{Malkin:Man'ko79} and polyatomic molecules in varying external fields,
crystals through which an electron is passing and exciting the oscillator
modes, and other interactions of the modes with external fields \cite%
{Fey:Hib}. Quadratic Hamiltonians have particular applications in quantum
electrodynamics because the electromagnetic field can be represented as a
set of forced harmonic oscillators \cite{Bo:Shi}, \cite{Dodonov:Man'koFIAN87}%
, \cite{Fey:Hib}, \cite{Merz} and \cite{Schiff}. Nonlinear oscillators play
a central role in the novel theory of Bose--Einstein condensation \cite%
{Dal:Giorg:Pitaevski:Str99}. From a general point of view, the dynamics of
gases of cooled atoms in a magnetic trap at very low temperatures can be
described by an effective equation for the condensate wave function known as
the Gross--Pitaevskii (or nonlinear Schr\"{o}dinger) equation \cite%
{Kagan:Surkov:Shlyap96}, \cite{Kagan:Surkov:Shlyap97}, \cite%
{Kivsh:Alex:Tur01} and \cite{Per-G:Tor:Mont}.\medskip\ 

In this Letter we consider the one-dimensional time-dependent Schr\"{o}%
dinger equation%
\begin{equation}
i\frac{\partial \psi }{\partial t}=H\psi  \label{Schroedinger}
\end{equation}%
with general variable quadratic Hamiltonians of the form%
\begin{equation}
H=a\left( t\right) p^{2}+b\left( t\right) x^{2}+c\left( t\right) px+d\left(
t\right) xp,\quad p=-i\frac{\partial }{\partial x},  \label{GenHam}
\end{equation}%
where $a\left( t\right) ,$ $b\left( t\right) ,$ $c\left( t\right) ,$ and $%
d\left( t\right) $ are real-valued functions of time, $t,$ only (see, for
example, \cite{Cor-Sot:Lop:Sua:Sus}, \cite{Cor-Sot:Sua:Sus}, \cite%
{Cor-Sot:Sua:SusInv}, \cite{Cor-Sot:Sus}, \cite{Dod:Mal:Man75}, \cite%
{FeynmanPhD}, \cite{Feynman}, \cite{Fey:Hib}, \cite{Lan:Sus}, \cite{Lop:Sus}%
, \cite{Me:Co:Su}, \cite{SuazoF}, \cite{Sua:Sus}, \cite{Suaz:Sus}, \cite%
{Sua:Sus:Vega}, \cite{Wolf81}, and \cite{Yeon:Lee:Um:George:Pandey93} for a
general approach and some elementary solutions). The corresponding Green
function, or Feynman's propagator, can be found as follows \cite%
{Cor-Sot:Lop:Sua:Sus}, \cite{Suaz:Sus}:%
\begin{equation}
\psi =G\left( x,y,t\right) =\frac{1}{\sqrt{2\pi i\mu _{0}\left( t\right) }}\
e^{i\left( \alpha _{0}\left( t\right) x^{2}+\beta _{0}\left( t\right)
xy+\gamma _{0}\left( t\right) y^{2}\right) },  \label{in2}
\end{equation}%
where%
\begin{eqnarray}
&&\alpha _{0}\left( t\right) =\frac{1}{4a\left( t\right) }\frac{\mu
_{0}^{\prime }\left( t\right) }{\mu _{0}\left( t\right) }-\frac{c\left(
t\right) }{2a\left( t\right) },  \label{in3} \\
&&\beta _{0}\left( t\right) =-\frac{\lambda \left( t\right) }{\mu _{0}\left(
t\right) },\qquad \lambda \left( t\right) =\exp \left( \int_{0}^{t}\left(
c\left( s\right) -d\left( s\right) \right) \ ds\right) ,  \label{in4} \\
&&\gamma _{0}\left( t\right) =\frac{a\left( t\right) \lambda ^{2}\left(
t\right) }{\mu _{0}\left( t\right) \mu _{0}^{\prime }\left( t\right) }+\frac{%
c\left( 0\right) }{2a\left( 0\right) }-4\int_{0}^{t}\frac{a\left( s\right)
\sigma \left( s\right) \lambda ^{2}\left( s\right) }{\left( \mu _{0}^{\prime
}\left( s\right) \right) ^{2}}\ ds,  \label{in5}
\end{eqnarray}%
and the function $\mu _{0}\left( t\right) $ satisfies the\ characteristic
equation%
\begin{equation}
\mu ^{\prime \prime }-\tau \left( t\right) \mu ^{\prime }+4\sigma \left(
t\right) \mu =0  \label{in6}
\end{equation}%
with%
\begin{equation}
\tau \left( t\right) =\frac{a^{\prime }}{a}+2c-2d,\qquad \sigma \left(
t\right) =ab-cd+\frac{c}{2}\left( \frac{a^{\prime }}{a}-\frac{c^{\prime }}{c}%
\right)  \label{in7}
\end{equation}%
subject to the initial data%
\begin{equation}
\mu _{0}\left( 0\right) =0,\qquad \mu _{0}^{\prime }\left( 0\right)
=2a\left( 0\right) \neq 0.  \label{in8}
\end{equation}%
(More details can be found in Refs.~\cite{Cor-Sot:Lop:Sua:Sus} and \cite%
{Suaz:Sus}.) Then by the superposition principle the solution of the Cauchy
initial value problem can be presented in an integral form%
\begin{equation}
\psi \left( x,t\right) =\int_{-\infty }^{\infty }G\left( x,y,t\right) \
\varphi \left( y\right) \ dy,\quad \lim_{t\rightarrow 0^{+}}\psi \left(
x,t\right) =\varphi \left( x\right)  \label{CauchyInVProb}
\end{equation}%
for a suitable initial function $\varphi $ on $%
\mathbb{R}
$ (a rigorous proof is given in \cite{Suaz:Sus} and uniqueness is analyzed
in \cite{Cor-Sot:Sua:SusInv}; another forms of solution are provided by (\ref%
{KSolution}) and (\ref{QIExpSolution})).\medskip

A detailed review on dynamical symmetries and quantum integrals of motion
for the time-dependent Schr\"{o}dinger equation can be found in \cite%
{Dodonov:Man'koFIAN87} and \cite{Malkin:Man'ko79} (see also an extensive
list of references therein). In this Letter, which is a continuation of the
recent paper \cite{Cor-Sot:Sua:SusInv}, a natural connection between the
linear and quadratic integrals of motion for general variable quadratic
Hamiltonians is established. As a result a nonlinear superposition principle
for the corresponding Ermakov systems, known as Pinney's solution, is
obtained and generalized. We pay also special attention to fundamental
relations between the linear dynamical invariants of Dodonov, Malkin, Man'ko
and Trifonov and solutions of the Cauchy initial value problem (see original
works \cite{Dod:Mal:Man75}, \cite{Dodonov:Man'koFIAN87}, \cite%
{Malkin:Man'ko79}, \cite{Malk:Man:Trif73}). In addition this initial value
problem is explicitly solved in terms of the quadratic invariant
eigenfunction expansion, which seems to be missing in the available
literature in general.\medskip

The Letter is organized as follows. We start from the standard definitions
of the dynamical symmetry and quantum integrals of motion, introducing also
some elementary tools, in the next two sections. Then in Section~4 we
describe all linear dynamical invariants for variable quadratic Hamiltonians
and determine their actions on the solutions of the corresponding
time-dependent Schr\"{o}dinger equation. In Section~5 all quadratic quantum
integrals of motion are characterized in terms of solutions of the
generalized Ermakov equation and a detailed proof is given. Their
decomposition into products of the linear invariants is derived in Section~6
and explicit actions on the solutions are established in Section~7. The last
section deals with a related nonlinear superposition principle for Ermakov's
equations and some computational details are provided in Appendices~A and B.

\section{Dynamical Symmetry}

In this Letter we elaborate on the following property.

\begin{lemma}
If%
\begin{equation}
i\frac{\partial \psi }{\partial t}=H\psi ,\qquad i\frac{\partial O}{\partial
t}+OH-H^{\dagger }O=0,  \label{Invariant}
\end{equation}%
then the function $\psi _{1}=O\psi $ satisfies the time-dependent Schr\"{o}%
dinger equation%
\begin{equation}
i\frac{\partial \psi _{1}}{\partial t}=H^{\dagger }\psi _{1},
\label{AdjointSchroedinger}
\end{equation}%
where $H^{\dagger }$ is the adjoint of the Hamiltonian $H.$
\end{lemma}

When $H=H^{\dagger },$ this property is taken as a definition of the
dynamical symmetry of the time-dependent Schr\"{o}dinger equation (\ref%
{Schroedinger}) (see, for example, \cite{Dodonov:Man'koFIAN87}, \cite%
{Lewis:Riesen69}, \cite{Malkin:Man'ko79} and references therein). At the
same time one has to deal with non-self-adjoint Hamiltonians in the theory
of dissipative quantum systems (see, for example, \cite{Cor-Sot:Sua:Sus}, 
\cite{Dekker81}, \cite{Kochan10}, \cite{Tarasov01}, \cite{Um:Yeon:George}
and references therein) or when using separation of variables in an
accelerating frame of reference for a charged particle moving in an uniform
variable magnetic field \cite{Cor-Sot:Lop:Sua:Sus}.

\begin{proof}
Partial differentiation%
\begin{eqnarray}
&&i\frac{\partial \psi _{1}}{\partial t}=i\frac{\partial }{\partial t}\left(
O\psi \right) =i\frac{\partial O}{\partial t}\psi +iO\frac{\partial \psi }{%
\partial t} \\
&&\qquad \quad =\left( H^{\dagger }O-OH\right) \psi +OH\psi =H^{\dagger
}\psi _{1}  \notag
\end{eqnarray}%
provides a direct proof.
\end{proof}

Definition of the dynamical symmetry is usually given in terms of solutions
of the same equation. A simple modification helps with the non-self-adjoint
quadratic Hamiltonians.

\begin{lemma}
The wave functions $\psi $ and $\chi ,$ related by%
\begin{equation}
\psi =\left( e^{-\int_{0}^{t}\left( c-d\right) ds}\ O\right) \chi ,
\label{Conjugate}
\end{equation}%
are solutions of the same Schr\"{o}dinger equation (\ref{Schroedinger})--(%
\ref{GenHam}), if the operator $O$ satisfies hypothesis of Lemma~1.
\end{lemma}

\begin{proof}
The simplest dynamical invariant, or an operator with the property (\ref%
{Invariant}), is given by%
\begin{equation}
O_{0}=O_{0}\left( c,d\right) =e^{\int_{0}^{t}\left( c-d\right) ds}\ I,
\label{SimplestInvariant}
\end{equation}%
where $I=id$ is the identity operator. (More details are provided in
Section~4.) Apply Lemma~1 twice in the following order%
\begin{equation}
\psi =O_{0}\left( d,c\right) \left( O\chi \right)
\end{equation}%
and use $\left( H^{\dagger }\right) ^{\dagger }=H$ to complete the proof.
\end{proof}

Examples occur throughout the Letter.

\section{Differentiation of Operators and Dynamical Invariants}

Following Lemma~1 we define the time derivative of an operator $O$ as follows%
\begin{equation}
\frac{dO}{dt}=\frac{\partial O}{\partial t}+\frac{1}{i}\left( OH-H^{\dagger
}O\right) ,  \label{Diff}
\end{equation}%
where $H^{\dagger }$ is the adjoint of the Hamiltonian operator $H.$ (This
formula is a simple extension of the well-known expression \cite{Deb:Mikus}, 
\cite{La:Lif}, \cite{Merz}, \cite{Schiff} to the case of a non-self-adjoint
Hamiltonian \cite{Cor-Sot:Sua:Sus}.) By definition%
\begin{equation}
\frac{dO}{dt}=\frac{\partial O}{\partial t}+\frac{1}{i}\left( OH-H^{\dagger
}O\right) =0  \label{DynamicalInvariant}
\end{equation}%
for any dynamical invariant. (In what follows we assume that the invariant $%
O $ does not involve time differentiation.)\medskip

This derivative is a linear operator%
\begin{equation}
\frac{d}{dt}\left( c_{1}O_{1}+c_{2}O_{2}\right) =c_{1}\frac{dO_{1}}{dt}+c_{2}%
\frac{dO_{2}}{dt}  \label{DiffLinear}
\end{equation}%
and the product rule takes the form%
\begin{eqnarray}
\frac{d}{dt}\left( O_{1}O_{2}\right) &=&\frac{\partial \left(
O_{1}O_{2}\right) }{\partial t}+\frac{1}{i}\left( \left( O_{1}O_{2}\right)
H-H^{\dagger }\left( O_{1}O_{2}\right) \right)  \label{ProductRule} \\
&=&\frac{dO_{1}}{dt}O_{2}+O_{1}\frac{dO_{2}}{dt}+iO_{1}\left( H-H^{\dag
}\right) O_{2}.  \notag
\end{eqnarray}%
For the general quadratic Hamiltonian, (\ref{GenHam}), one gets%
\begin{equation}
\frac{d}{dt}\left( O_{1}O_{2}\right) =\frac{dO_{1}}{dt}O_{2}+O_{1}\frac{%
dO_{2}}{dt}+\left( c-d\right) O_{1}O_{2}  \label{ProductRuleA}
\end{equation}%
and by definition (\ref{Diff})%
\begin{eqnarray}
\frac{d}{dt}\left( e^{\alpha \int_{0}^{t}\left( c-d\right) ds}\
O_{1}O_{2}\right) &=&e^{\alpha \int_{0}^{t}\left( c-d\right) ds}\ \left( 
\frac{dO_{1}}{dt}O_{2}+O_{1}\frac{dO_{2}}{dt}\right) \\
&&+\left( \alpha +1\right) \left( c-d\right) e^{\alpha \int_{0}^{t}\left(
c-d\right) ds}\ O_{1}O_{2}.  \notag
\end{eqnarray}%
If $\alpha =-1,$ we finally obtain%
\begin{equation}
\frac{d}{dt}\left( e^{-\int_{0}^{t}\left( c-d\right) ds}\ O_{1}O_{2}\right)
=e^{-\int_{0}^{t}\left( c-d\right) ds}\ \left( \frac{dO_{1}}{dt}O_{2}+O_{1}%
\frac{dO_{2}}{dt}\right) .  \label{productrule}
\end{equation}%
This implies that, if the operators $O_{1}$ and $O_{2}$ are dynamical
invariants, namely,%
\begin{eqnarray}
&&\frac{dO_{1}}{dt}=\frac{\partial O_{1}}{\partial t}+\frac{1}{i}\left(
O_{1}H-H^{\dagger }O_{1}\right) =0, \\
&&\frac{dO_{2}}{dt}=\frac{\partial O_{2}}{\partial t}+\frac{1}{i}\left(
O_{2}H-H^{\dagger }O_{2}\right) =0,
\end{eqnarray}%
then their modified product,%
\begin{equation}
E=e^{-\int_{0}^{t}\left( c-d\right) ds}\ O_{1}O_{2},
\label{dynamicinvariant}
\end{equation}%
is also a dynamical invariant:%
\begin{equation}
\frac{dE}{dt}=\frac{\partial E}{\partial t}+\frac{1}{i}\left( EH-H^{\dagger
}E\right) =0.
\end{equation}%
In Section~6 this property allows us to describe connections between linear
and quadratic dynamical invariants of the time-dependent Hamiltonian (\ref%
{GenHam}).

\section{Linear Integrals of Motion}

All invariants of the form%
\begin{equation}
P=A\left( t\right) p+B\left( t\right) x+C\left( t\right) ,\quad \frac{dP}{dt}%
=0,  \label{LinInts}
\end{equation}%
(we call them the Dodonov--Malkin--Man'ko--Trifonov invariants; see, for
example, \cite{Dod:Mal:Man75}, \cite{Dodonov:Man'koFIAN87}, \cite%
{Malkin:Man'ko79}, and \cite{Malk:Man:Trif73} and references therein) for
the general variable quadratic Hamiltonian (\ref{GenHam}) can be found in
the following fashion. Use of the differentiation formula (\ref{Diff})
results in the system \cite{Cor-Sot:Sua:SusInv}:%
\begin{eqnarray}
A^{\prime } &=&2c\left( t\right) A-2a\left( t\right) B,  \label{LinCA} \\
B^{\prime } &=&2b\left( t\right) A-2d\left( t\right) B,  \label{LinCB} \\
C^{\prime } &=&\left( c\left( t\right) -d\left( t\right) \right) C.
\label{LinCC}
\end{eqnarray}%
The last equation is explicitly integrated and elimination of $B$ and $A$
from (\ref{LinCA}) and (\ref{LinCB}), respectively, gives the second-order
equations:%
\begin{eqnarray}
A^{\prime \prime }-\left( \frac{a^{\prime }}{a}+2c-2d\right) A^{\prime
}+4\left( ab-cd+\frac{c}{2}\left( \frac{a^{\prime }}{a}-\frac{c^{\prime }}{c}%
\right) \right) A &=&0,  \label{EqC} \\
B^{\prime \prime }-\left( \frac{b^{\prime }}{b}+2c-2d\right) B^{\prime
}+4\left( ab-cd-\frac{d}{2}\left( \frac{b^{\prime }}{b}-\frac{d^{\prime }}{d}%
\right) \right) B &=&0.  \label{EqBC}
\end{eqnarray}%
The first is simply the characteristic equation (\ref{in6})--(\ref{in7}).
(It also coincides with the Ehrenfest theorem (\cite{Cor-Sot:Sua:SusInv}, 
\cite{Ehrenfest}) when $c\leftrightarrow d.)$\medskip

Thus all linear quantum invariants are given by%
\begin{equation}
P=A\left( t\right) p+\frac{2c\left( t\right) A\left( t\right) -A^{\prime
}\left( t\right) }{2a\left( t\right) }x+C_{0}\exp \left( \int_{0}^{t}\left(
c\left( s\right) -d\left( s\right) \right) \ ds\right) ,  \label{GenLinInv}
\end{equation}%
where $A\left( t\right) $ is a general solution of equation (\ref{EqC})
depending upon two parameters and $C_{0}$ is the third constant. Study of
spectra of the linear dynamical invariants allows one to solve the Cauchy
initial value problem (\cite{Dod:Mal:Man75}, \cite{Dodonov:Man'koFIAN87}, 
\cite{Malkin:Man'ko79} and \cite{Malk:Man:Trif73}).

\begin{theorem}
\textbf{(Eigenvalue Problem for the Linear Invariants)} If%
\begin{equation}
P\left( t\right) =\mu \left( t\right) p+\frac{2c\left( t\right) \mu \left(
t\right) -\mu ^{\prime }\left( t\right) }{2a\left( t\right) }x,
\label{LinMomentum}
\end{equation}%
then for any solution $A=\mu \left( t\right) $ of the characteristic
equation (\ref{in6})--(\ref{in7}) we have%
\begin{equation}
P\left( t\right) K\left( x,y,t\right) =\beta \left( 0\right) \mu \left(
0\right) \lambda \left( t\right) y\ K\left( x,y,t\right) .
\label{EigenValueProblem}
\end{equation}%
The eigenfunctions are bounded solutions of the time-dependent Schr\"{o}%
dinger equation (\ref{Schroedinger}) given by%
\begin{equation}
K\left( x,y,t\right) =\frac{1}{\sqrt{2\pi \mu \left( t\right) }}\ e^{i\left(
\alpha \left( t\right) x^{2}+\beta \left( t\right) xy+\gamma \left( t\right)
y^{2}\right) },  \label{KKernel}
\end{equation}%
where $\lambda \left( t\right) =\exp \left( \int_{0}^{t}\left( c\left(
s\right) -d\left( s\right) \right) \ ds\right) $ and%
\begin{eqnarray}
\mu \left( t\right) &=&2\mu \left( 0\right) \mu _{0}\left( t\right) \left(
\alpha \left( 0\right) +\gamma _{0}\left( t\right) \right) ,  \label{MKernel}
\\
\alpha \left( t\right) &=&\frac{1}{4a\left( t\right) }\frac{\mu ^{\prime
}\left( t\right) }{\mu \left( t\right) }-\frac{c\left( t\right) }{2a\left(
t\right) }  \label{AKernel} \\
&=&\alpha _{0}\left( t\right) -\frac{\beta _{0}^{2}\left( t\right) }{4\left(
\alpha \left( 0\right) +\gamma _{0}\left( t\right) \right) },  \notag \\
\beta \left( t\right) &=&-\frac{\beta \left( 0\right) \beta _{0}\left(
t\right) }{2\left( \alpha \left( 0\right) +\gamma _{0}\left( t\right)
\right) }  \label{BKernel} \\
&=&\frac{\beta \left( 0\right) \mu \left( 0\right) }{\mu \left( t\right) }%
\lambda \left( t\right) ,  \notag \\
\gamma \left( t\right) &=&\gamma \left( 0\right) -\frac{\beta ^{2}\left(
0\right) }{4\left( \alpha \left( 0\right) +\gamma _{0}\left( t\right)
\right) }.  \label{CKernel}
\end{eqnarray}%
(When $\mu =\mu _{0}$ with $\mu _{0}\left( 0\right) =0$ and $\mu
_{0}^{\prime }\left( 0\right) =2a\left( 0\right) \neq 0,$ we obtain the
Green function (\ref{in2})--(\ref{in5}) up to a simple factor.)
\end{theorem}

\begin{proof}
The required solution (\ref{KKernel})--(\ref{CKernel}) has been already
constructed in \cite{Suaz:Sus} and one has to verify formula (\ref%
{EigenValueProblem}) only. We have%
\begin{eqnarray}
PK\left( x,y,t\right) &=&\left( Ap+Bx\right) K\left( x,y,t\right) \\
&=&\left( \left( 2\alpha A+B\right) x+\beta Ay\right) K\left( x,y,t\right) ,
\notag
\end{eqnarray}%
where%
\begin{eqnarray}
2\alpha A+B &=&\left( \frac{1}{2a}\frac{\mu ^{\prime }}{\mu }-\frac{c}{a}%
\right) A+\frac{2cA-A^{\prime }}{2a} \\
&=&\frac{A}{2a}\left( \frac{\mu ^{\prime }}{\mu }-\frac{A^{\prime }}{A}%
\right) =0  \notag
\end{eqnarray}%
provided $\mu =A.$ Then $\beta A=\beta \left( 0\right) \mu \left( 0\right)
\lambda $ and the proof is complete.
\end{proof}

Our theorem can be thought as a natural extension to the case of
non-self-adjoint variable quadratic Hamiltonians (\ref{GenHam}) of a
familiar relation between the Green function and linear dynamical invariants
established in \cite{Dod:Mal:Man75}, \cite{Dodonov:Man'koFIAN87}, \cite%
{Malkin:Man'ko79} and \cite{Malk:Man:Trif73}. The time-dependent factor in
the eigenvalue (\ref{EigenValueProblem}) corresponds to the statement of
Lemma~2.\medskip\ 

In this Letter we are interested, among other things, in a direct
verification of Lemma~1 for linear and quadratic dynamical invariants. For
the general Dodonov--Malkin--Man'ko--Trifonov invariant (\ref{GenLinInv}),
without any loss of generality, one can separately consider two cases, say,
when $A\left( t\right) \equiv 0$ with $C_{0}=1$ and when $C_{0}=0.$ If%
\begin{equation}
\psi _{1}=\psi \ e^{\int_{0}^{t}\left( c-d\right) ds},
\end{equation}%
then%
\begin{eqnarray}
&&i\frac{\partial \psi _{1}}{\partial t}=i\frac{\partial \psi }{\partial t}\
e^{\int_{0}^{t}\left( c-d\right) ds}+i\left( c-d\right) \psi _{1} \\
&&\qquad \quad =\left( ap^{2}+bx^{2}+cpx+dxp\right) \psi _{1}  \notag \\
&&\qquad \qquad +\left( c-d\right) \left( xp-px\right) \psi _{1}=H^{\dagger
}\psi _{1},
\end{eqnarray}%
which takes care of the first case (we have verified once again the
statement of Lemma~1 for the simplest invariant (\ref{SimplestInvariant}%
)).\medskip

In the second case, $C_{0}=0,$ we follow Theorem~1 and take a solution $%
A=\mu \left( t\right) $ of (\ref{EqC}) and (\ref{in6}), which does not have
to satisfy the initial conditions of the Green function (\ref{in8}). Then by
the superposition principle solution of the corresponding Cauchy initial
values problem is given by the integral operator \cite{Suaz:Sus}%
\begin{equation}
\psi \left( x,t\right) =\int_{-\infty }^{\infty }K\left( x,y,t\right) \ \chi
\left( y\right) \ dy,\qquad \psi \left( x,0\right) =\int_{-\infty }^{\infty
}K\left( x,y,0\right) \ \chi \left( y\right) \ dy  \label{KSolution}
\end{equation}%
with the kernel (\ref{KKernel})--(\ref{CKernel}). Thus%
\begin{equation}
\psi _{1}\left( x,t\right) =P\psi \left( x,t\right) =\int_{-\infty }^{\infty
}\left( PK\left( x,y,t\right) \right) \ \chi \left( y\right) \ dy
\label{PKSolution}
\end{equation}%
(we freely interchange differentiation and integration in this Letter; it
can be justified for certain classes of solutions (\cite{Lieb:Loss}, \cite%
{Oh89}, \cite{Per-G:Tor:Mont}, \cite{Velo96})). By choosing $\beta \left(
0\right) \mu \left( 0\right) =1$ in (\ref{EigenValueProblem}) we obtain

\begin{equation}
\psi _{1}\left( x,t\right) =P\psi \left( x,t\right) =e^{\int_{0}^{t}\left(
c-d\right) ds}\ \int_{-\infty }^{\infty }K\left( x,y,t\right) \ \left( y\chi
\left( y\right) \right) \ dy,  \label{SecondKSolution}
\end{equation}%
where the second factor obviously satisfies the Schr\"{o}dinger equation (%
\ref{Schroedinger}) (see \cite{Suaz:Sus} for details). Repeating the first
step one completes the proof. Eq.~(\ref{SecondKSolution}) provides a
spectral decomposion (\cite{Reed:SimonI}, \cite{Reed:SimonII}) of the linear
invariant $P,$ which has a continuous spectrum.\medskip

Our consideration shows how Lemma~1 works for the general
Dodonov--Malkin--Man'ko--Trifonov invariant --- application of this
invariant to a given solution of the corresponding Cauchy initial value
problem simply produces another solution with the following initial data:%
\begin{eqnarray}
\psi _{1}\left( x,0\right) &=&P\left( 0\right) \psi \left( x,0\right)
\label{InitialData} \\
&=&\int_{-\infty }^{\infty }K\left( x,y,0\right) \ \left( y\chi \left(
y\right) \right) \ dy.  \notag
\end{eqnarray}%
The reader can easily connect initial conditions of two solutions (\ref%
{KSolution}) and (\ref{SecondKSolution}) with the help of an analog of the
Fourier transform (see also Ref.~\cite{Sua:Sus} for a formal inversion
operator in general).

\section{Quadratic Dynamical Invariants}

All quantum quadratic integrals of motion are given by%
\begin{equation}
E=A\left( t\right) p^{2}+B\left( t\right) x^{2}+C\left( t\right) \left(
px+xp\right) ,\quad \frac{dE}{dt}=0,  \label{QuadraticInvariant}
\end{equation}%
and can be found as follows \cite{Cor-Sot:Sua:SusInv}.

\begin{lemma}
The quadratic dynamical invariants of the Hamiltonian (\ref{GenHam}) have
the form%
\begin{eqnarray}
E\left( t\right) &=&\left[ \left( \kappa \ p+\frac{\left( c+d\right) \kappa
-\kappa ^{\prime }}{2a}\ x\right) ^{2}+\frac{C_{0}}{\kappa ^{2}}\ x^{2}%
\right]  \label{InvSymmForm} \\
&&\times \exp \left( \int_{0}^{t}\left( c-d\right) \ ds\right) ,  \notag
\end{eqnarray}%
where $C_{0}$ is a constant and function $\kappa \left( t\right) $ is a
solution of the following nonlinear auxiliary equation:%
\begin{equation}
\kappa ^{\prime \prime }-\frac{a^{\prime }}{a}\kappa ^{\prime }+\left[
4ab+\left( \frac{a^{\prime }}{a}-c-d\right) \left( c+d\right) -c^{\prime
}-d^{\prime }\right] \kappa =C_{0}\frac{\left( 2a\right) ^{2}}{\kappa ^{3}}.
\label{AuxEquation}
\end{equation}
\end{lemma}

The structure of quadratic invariants is once again in perfect agreement
with Lemma~2.

\begin{proof}
In a general form this result has been established in \cite%
{Cor-Sot:Sua:SusInv} (see also Refs.~\cite{Lewis:Riesen69}, \cite{Symon70}
and \cite{Yeon:Lee:Um:George:Pandey93} for important earlier articles). A
somewhat different and more direct proof is presented here. It is sufficient
to show that the corresponding linear system%
\begin{eqnarray}
A^{\prime }+4aC-\left( 3c+d\right) A &=&0,  \label{EqA} \\
B^{\prime }-4bC+\left( c+3d\right) B &=&0,  \label{EqB} \\
C^{\prime }+2\left( aB-bA\right) -\left( c-d\right) C &=&0  \label{EqCOne}
\end{eqnarray}%
has the following general solution%
\begin{eqnarray}
A\left( t\right) &=&\kappa ^{2}\exp \left( \int_{0}^{t}\left( c-d\right) \
ds\right) ,  \label{SolA} \\
B\left( t\right) &=&\left[ \left( \frac{\kappa ^{\prime }-\left( c+d\right)
\kappa }{2a}\right) ^{2}+\frac{C_{0}}{\kappa ^{2}}\right] \exp \left(
\int_{0}^{t}\left( c-d\right) \ ds\right) ,  \label{SolB} \\
C\left( t\right) &=&\kappa \frac{\left( c+d\right) \kappa -\kappa ^{\prime }%
}{2a}\exp \left( \int_{0}^{t}\left( c-d\right) \ ds\right) ,  \label{SolC}
\end{eqnarray}%
where $C_{0}$ is a constant and the function $\kappa \left( t\right) $
satisfies the nonlinear auxiliary equation (\ref{AuxEquation}).\medskip

Indeed the substitution%
\begin{equation}
A\left( t\right) =\widetilde{A}\left( t\right) \lambda \left( t\right)
,\quad B\left( t\right) =\widetilde{B}\left( t\right) \lambda \left(
t\right) ,\quad \quad C\left( t\right) =\widetilde{C}\left( t\right) \lambda
\left( t\right) ,
\end{equation}%
where $\lambda \left( t\right) =\exp \left( \int_{0}^{t}\left( c-d\right) \
ds\right) ,$ transforms the original system into a more convenient form%
\begin{eqnarray}
\widetilde{A}^{\prime }+4a\widetilde{C}-2\left( c+d\right) \widetilde{A}
&=&0,  \label{EqAT} \\
\widetilde{B}^{\prime }-4b\widetilde{C}+2\left( c+d\right) \widetilde{B}
&=&0,  \label{EqBT} \\
\widetilde{C}^{\prime }+2\left( a\widetilde{B}-b\widetilde{A}\right) &=&0.
\label{EqCT}
\end{eqnarray}%
Letting%
\begin{equation}
\widetilde{A}=\kappa ^{2}\quad \text{and\quad }\widetilde{A}^{\prime
}=2\kappa \kappa ^{\prime }  \label{ATildeKappa}
\end{equation}%
in the first equation (\ref{EqAT}), we obtain%
\begin{eqnarray}
\widetilde{C} &=&\frac{\kappa }{2a}\left( \left( c+d\right) \kappa -\kappa
^{\prime }\right)  \label{CTildeKappa} \\
&=&-e^{\int \left( c+d\right) dt}\ \frac{\kappa }{2a}\frac{d}{dt}\left(
\kappa e^{-\int \left( c+d\right) dt}\right) .  \notag
\end{eqnarray}%
Then from the third equation (\ref{EqCT}):%
\begin{eqnarray}
\widetilde{B} &=&\frac{b}{a}\kappa ^{2}+\frac{1}{2a}\left[ \kappa \frac{%
\kappa ^{\prime }-\left( c+d\right) \kappa }{2a}\right] ^{\prime }
\label{BTildeKappa} \\
&=&\frac{b}{a}\kappa ^{2}+\frac{1}{2a}\frac{d}{dt}\left[ e^{\int \left(
c+d\right) dt}\ \frac{\kappa }{2a}\frac{d}{dt}\left( \kappa e^{-\int \left(
c+d\right) dt}\right) \right]  \notag
\end{eqnarray}%
and substitution of (\ref{CTildeKappa})--(\ref{BTildeKappa}) into (\ref{EqBT}%
) gives%
\begin{equation}
k\frac{d}{dt}\left[ 4abk^{2}\mu ^{2}+k\frac{d}{dt}\left( \mu \left( k\frac{%
d\mu }{dt}\right) \right) \right] +8abk^{2}\mu \left( k\frac{d\mu }{dt}%
\right) =0  \label{EqKappa}
\end{equation}%
in the following temporary notations%
\begin{equation}
\mu =\kappa e^{-\int \left( c+d\right) dt},\qquad k=\frac{1}{2a}e^{2\int
\left( c+d\right) dt}.  \label{NewNotations}
\end{equation}%
Using%
\begin{equation}
\mu \left( t\right) =y\left( \tau \right) ,\quad k\frac{d\mu }{dt}=\frac{dy}{%
d\tau },\quad q=4abk^{2}=\frac{b}{a}e^{4\int \left( c+d\right) dt},
\label{NewNNotations}
\end{equation}%
one gets%
\begin{equation}
\frac{d}{d\tau }\left[ qy^{2}+\frac{d}{d\tau }\left( y\frac{dy}{d\tau }%
\right) \right] +2qy\frac{dy}{d\tau }=0
\end{equation}%
or%
\begin{equation}
y\left( y^{\prime \prime }+qy\right) ^{\prime }+3y^{\prime }\left( y^{\prime
\prime }+qy\right) =0
\end{equation}%
(see also Ref.~\cite{Lewis:Riesen69} for an important special case). By an
integrating factor:%
\begin{equation}
\frac{d}{d\tau }\left[ y^{3}\left( y^{\prime \prime }+qy\right) \right]
=0,\quad y^{\prime \prime }+qy=\frac{C_{0}}{y^{3}},  \label{TauIntegral}
\end{equation}%
and a back substitution with the help of%
\begin{equation}
\frac{d^{2}y}{d\tau ^{2}}=\frac{e^{3\int \left( c+d\right) dt}}{\left(
2a\right) ^{2}}\ \left[ \kappa ^{\prime \prime }-\frac{a^{\prime }}{a}\kappa
^{\prime }+\left( \left( \frac{a^{\prime }}{a}-c-d\right) \left( c+d\right)
-c^{\prime }-d^{\prime }\right) \kappa \right]
\end{equation}%
results in the required first integral of the system:%
\begin{equation}
\frac{d}{dt}\left[ \frac{\kappa ^{3}}{\left( 2a\right) ^{2}}\left( \kappa
^{\prime \prime }-\frac{a^{\prime }}{a}\kappa ^{\prime }+\left( 4ab+\left( 
\frac{a^{\prime }}{a}-c-d\right) \left( c+d\right) -c^{\prime }-d^{\prime
}\right) \kappa \right) \right] =0,
\end{equation}%
which gives our auxiliary equation (\ref{AuxEquation}).\medskip

The last term in (\ref{BTildeKappa}) can be transformed as follows%
\begin{eqnarray*}
&&\frac{1}{2a}\frac{d}{dt}\left[ e^{\int \left( c+d\right) dt}\ \frac{\kappa 
}{2a}\frac{d}{dt}\left( \kappa e^{-\int \left( c+d\right) dt}\right) \right]
\\
&&\qquad =e^{-2\int \left( c+d\right) dt}\ \frac{d}{d\tau }\left( y\frac{dy}{%
d\tau }\right) =e^{-2\int \left( c+d\right) dt}\ \left( yy^{\prime \prime
}+\left( y^{\prime }\right) ^{2}\right) \\
&&\qquad =e^{-2\int \left( c+d\right) dt}\ \left[ \left( \frac{dy}{d\tau }%
\right) ^{2}-qy^{2}+\frac{C_{0}}{y^{2}}\right] \\
&&\qquad =\left( \frac{\kappa ^{\prime }-\left( c+d\right) \kappa }{2a}%
\right) ^{2}-\frac{b}{a}\kappa ^{2}+\frac{C_{0}}{\kappa ^{2}},
\end{eqnarray*}%
in view of (\ref{NewNotations})--(\ref{NewNNotations}) and (\ref{TauIntegral}%
). We have also utilized a convenient identity%
\begin{equation}
\frac{dy}{d\tau }=\frac{1}{2a}\left( \frac{d\kappa }{dt}-\left( c+d\right)
\kappa \right) e^{\int \left( c+d\right) dt}.
\end{equation}%
Thus%
\begin{equation}
\widetilde{B}=\left( \frac{\kappa ^{\prime }-\left( c+d\right) \kappa }{2a}%
\right) ^{2}+\frac{C_{0}}{\kappa ^{2}}  \label{BBTildeKappa}
\end{equation}%
and the proof is complete.
\end{proof}

The case $a=1/2,$ $b=\omega ^{2}\left( t\right) /2$ and $c=d=0$ corresponds
to the familiar Ermakov--Lewis--Riesenfeld invariant \cite{Ermakov}, \cite%
{Lewis67}, \cite{Lewis68}, \cite{Lewis68a}, \cite{Lewis:Riesen69}; see also
Refs.~\cite{Lewis:Leach82a} and \cite{Lewis:Leach82b}. (The corresponding
classical invariant in general is discussed in Refs.~\cite{Symon70} and \cite%
{Yeon:Lee:Um:George:Pandey93}.) Many examples of completely integrable
generalized harmonic oscillators and their integrals of motion are given in
Ref.~\cite{Cor-Sot:Sua:SusInv}.\medskip

For a positive constant, $C_{0}>0,$ the quantum dynamical invariant (\ref%
{InvSymmForm}) can be presented in the standard harmonic oscillator form 
\cite{Cor-Sot:Sua:Sus}, \cite{Cor-Sot:Sua:SusInv}, \cite{Reed:SimonI}, \cite%
{Reed:SimonII}: 
\begin{equation}
E=\frac{\omega \left( t\right) }{2}\left( \widehat{a}\left( t\right) 
\widehat{a}^{\dagger }\left( t\right) +\widehat{a}^{\dagger }\left( t\right) 
\widehat{a}\left( t\right) \right) ,  \label{EnOperFactor}
\end{equation}%
where%
\begin{equation}
\omega \left( t\right) =\omega _{0}\exp \left( \int_{0}^{t}\left( c-d\right)
\ ds\right) ,\qquad \omega _{0}=2\sqrt{C_{0}}>0,  \label{omega(t)}
\end{equation}%
\begin{eqnarray}
\widehat{a}\left( t\right) &=&\left( \frac{\sqrt{\omega _{0}}}{2\kappa }-i%
\frac{\kappa ^{\prime }-\left( c+d\right) \kappa }{2a\sqrt{\omega _{0}}}%
\right) x+\frac{\kappa }{\sqrt{\omega _{0}}}\frac{\partial }{\partial x},
\label{a(t)} \\
\widehat{a}^{\dagger }\left( t\right) &=&\left( \frac{\sqrt{\omega _{0}}}{%
2\kappa }+i\frac{\kappa ^{\prime }-\left( c+d\right) \kappa }{2a\sqrt{\omega
_{0}}}\right) x-\frac{\kappa }{\sqrt{\omega _{0}}}\frac{\partial }{\partial x%
},  \label{across(t)}
\end{eqnarray}%
and $\kappa $ is a real-valued solution of the nonlinear auxiliary equation (%
\ref{AuxEquation}). Here the time-dependent annihilation $\widehat{a}\left(
t\right) $ and creation $\widehat{a}^{\dagger }\left( t\right) $ operators
satisfy the canonical commutation relation:%
\begin{equation}
\widehat{a}\left( t\right) \widehat{a}^{\dagger }\left( t\right) -\widehat{a}%
^{\dagger }\left( t\right) \widehat{a}\left( t\right) =1.
\label{commutatora(t)across(t)}
\end{equation}%
The oscillator-type spectrum of the dynamical invariant $E$ can be obtained
now in a standard way by using the Heisenberg--Weyl algebra of the rasing
and lowering operators (a \textquotedblleft second
quantization\textquotedblright\ \cite{Lewis:Riesen69}, the Fock states):%
\begin{equation}
\widehat{a}\left( t\right) \Psi _{n}\left( x,t\right) =\sqrt{n}\ \Psi
_{n-1}\left( x,t\right) ,\quad \widehat{a}^{\dagger }\left( t\right) \Psi
_{n}\left( x,t\right) =\sqrt{n+1}\ \Psi _{n+1}\left( x,t\right) ,
\label{annandcratoperactions}
\end{equation}%
\begin{equation}
E\left( t\right) \Psi _{n}\left( x,t\right) =\omega \left( t\right) \left( n+%
\frac{1}{2}\right) \Psi _{n}\left( x,t\right) .  \label{Eeigenvp}
\end{equation}%
The corresponding orthonormal time-dependent eigenfunctions are given by%
\begin{equation}
\Psi _{n}\left( x,t\right) =\exp \left( i\frac{\kappa ^{\prime }-\left(
c+d\right) \kappa }{4a\kappa }x^{2}\right) v_{n}\left( \xi \right)
\label{Eeigenfs}
\end{equation}%
in terms of Hermite polynomials \cite{Ni:Su:Uv}: 
\begin{equation}
v_{n}=C_{n}e^{-\xi ^{2}/2}H_{n}\left( \xi \right) ,\qquad \xi =\varepsilon x,
\label{Eeigenfvs}
\end{equation}%
where%
\begin{equation}
\left\vert C_{n}\right\vert ^{2}=\frac{\varepsilon }{\sqrt{\pi }2^{n}n!}%
,\qquad \varepsilon ^{2}=\frac{\omega _{0}}{2\kappa ^{2}}=\frac{\sqrt{C_{0}}%
}{\kappa ^{2}}.  \label{Eeigenfsconstant}
\end{equation}%
Their relation with the original Cauchy initial value problem is discussed
in Section~7 (see Theorem~2). In addition the $n$-dimensional oscillator
wave functions form a basis of the irreducible unitary representation of the
Lie algebra of the noncompact group $SU\left( 1,1\right) $ corresponding to
the discrete positive series $\mathcal{D}_{+}^{j}$ (see \cite{Me:Co:Su}, 
\cite{Ni:Su:Uv} and \cite{Smir:Shit}).\smallskip\ Operators (\ref{a(t)})--(%
\ref{across(t)}) allow us to extend these group-theoretical properties to
the general quadratic dynamical invariant (\ref{EnOperFactor}).\medskip

When $C_{0}=0,$ the dynamical invariant (\ref{InvSymmForm}) is given as the
square of a linear invariant up to a simple factor, which resembles the case
of a free particle. If $C_{0}<0,$ one deals with the Hamiltonian of a linear
repulsive \textquotedblleft oscillator\textquotedblright .

\section{Relation Between Linear and Quadratic Invariants}

By Lemma~3 the operators $p^{2},$ $x^{2}$ and $px+xp$ form a basis for all
quadratic invariants. Here we take two linearly independent solutions, say $%
\mu _{1}=A_{1}$ and $\mu _{2}=A_{2},$ of equations (\ref{in6}) and (\ref{EqC}%
) and consider the corresponding Dodonov--Malkin--Man'ko--Trifonov
invariants (\ref{GenLinInv}): 
\begin{equation}
P_{1}=A_{1}p+B_{1}x,\qquad P_{2}=A_{2}p+B_{2}x.
\end{equation}%
Introducing the following quadratic invariants%
\begin{equation}
E_{1}=P_{1}^{2}\ e^{-\int_{0}^{t}\left( c-d\right) ds},\quad
E_{2}=P_{2}^{2}\ e^{-\int_{0}^{t}\left( c-d\right) ds},\quad E_{3}=\left(
P_{1}P_{2}+P_{2}P_{1}\right) \ e^{-\int_{0}^{t}\left( c-d\right) ds}
\end{equation}%
with the help of (\ref{dynamicinvariant}) as another basis, one gets%
\begin{eqnarray}
E &=&C_{1}E_{1}+C_{2}E_{2}+C_{3}E_{3}  \label{Linear&Quadratic} \\
&=&\left[ C_{1}P_{1}^{2}+C_{2}P_{2}^{2}+C_{3}\left(
P_{1}P_{2}+P_{2}P_{1}\right) \right] \exp \left( -\int_{0}^{t}\left(
c-d\right) \ ds\right)  \notag
\end{eqnarray}%
for some constants $C_{1},$ $C_{2}$ and $C_{3}.$ As a result the following
operator identity holds%
\begin{eqnarray}
&&\left[ \left( \kappa \ p+\frac{\left( c+d\right) \kappa -\kappa ^{\prime }%
}{2a}\ x\right) ^{2}+\frac{C_{0}}{\kappa ^{2}}\ x^{2}\right] \exp \left(
\int_{0}^{t}\left( c-d\right) \ ds\right)  \label{OpIdentity} \\
&&\quad =\left[ C_{1}\left( A_{1}p+B_{1}x\right) ^{2}+C_{2}\left(
A_{2}p+B_{2}x\right) ^{2}\right.  \notag \\
&&\qquad \left. +C_{3}\left( \left( A_{1}p+B_{1}x\right) \left(
A_{2}p+B_{2}x\right) +\left( A_{2}p+B_{2}x\right) \left(
A_{1}p+B_{1}x\right) \right) \right]  \notag \\
&&\qquad \qquad \times \exp \left( -\int_{0}^{t}\left( c-d\right) \
ds\right) ,  \notag
\end{eqnarray}%
where%
\begin{equation}
A_{1}=\mu _{1},\quad B_{1}=\frac{2c\mu _{1}-\mu _{1}^{\prime }}{2a},\qquad
A_{2}=\mu _{2},\quad B_{2}=\frac{2c\mu _{2}-\mu _{2}^{\prime }}{2a}.
\label{OpIdentityAB}
\end{equation}%
Thus we obtain%
\begin{equation}
\kappa ^{2}=\left( C_{1}\mu _{1}^{2}+C_{2}\mu _{2}^{2}+2C_{3}\mu _{1}\mu
_{2}\right) \exp \left( -2\int_{0}^{t}\left( c-d\right) \ ds\right)
\label{NonlinSuperposition}
\end{equation}%
as a relation between solutions of the nonlinear auxiliary equation (\ref%
{AuxEquation}) and the linear characteristic equation (\ref{EqC}). In
addition the substitution%
\begin{equation}
\mu _{1}=\kappa _{1}\exp \left( \int_{0}^{t}\left( c-d\right) \ ds\right)
,\qquad \mu _{2}=\kappa _{2}\exp \left( \int_{0}^{t}\left( c-d\right) \
ds\right)  \label{MuKappaSubstitution}
\end{equation}%
transforms the characteristic equation (\ref{EqC}) into our auxiliary
equation (\ref{AuxEquation}) with $C_{0}=0.$ Finally a general solution of
the nonlinear equation is given by the following \textquotedblleft operator
law of cosines\textquotedblright :%
\begin{equation}
\kappa ^{2}\left( t\right) =C_{1}\kappa _{1}^{2}\left( t\right) +C_{2}\kappa
_{2}^{2}\left( t\right) +2C_{3}\kappa _{1}\left( t\right) \kappa _{2}\left(
t\right)  \label{AuxSol}
\end{equation}%
in terms of two linearly independent solutions $\kappa _{1}$ and $\kappa
_{2} $ of the homogeneous equation. The constant $C_{0}$ is related to the
Wronskian of two linearly independent solutions $\kappa _{1}$ and $\kappa
_{2}:$ 
\begin{equation}
C_{1}C_{2}-C_{3}^{2}=C_{0}\frac{\left( 2a\right) ^{2}}{W^{2}\left( \kappa
_{1},\kappa _{2}\right) },\quad W\left( \kappa _{1},\kappa _{2}\right)
=\kappa _{1}\kappa _{2}^{\prime }-\kappa _{1}^{\prime }\kappa _{2}
\label{AuxSolWronskian}
\end{equation}%
(more details are given in Appendix~A). This is a well-known nonlinear
superposition property of the so-called Ermakov systems (see, for example, 
\cite{Corant:Snyder58}, \cite{Eliezer:Gray76}, \cite{Ermakov}, \cite%
{Leach:Andrio08}, \cite{Leach:Andriopo08}, \cite{Leach:Kar:Nuc:Andrio05}, 
\cite{Lewis68a}, \cite{Mah:Leach07}, \cite{Nucci:Leach05}, \cite%
{PadillaMaster}, \cite{Pinney50}, \cite{Schuch08} and references therein).
Here we have obtained this \textquotedblleft nonlinear superposition
principle\textquotedblright\ (or Pinney's solution) in an operator form by
multiplication and addition of the linear dynamical invariants together with
an independent characterization of all quantum quadratic invariants, which
seems to be missing, in general, in the available literature (see also \cite%
{Lewis67} and \cite{Lewis68a} for an important classical case). An extension
is given in the last section.\medskip

It is worth noting, in conclusion, that the linear invariants of Dodonov,
Malkin, Man'ko and Trifonov (\cite{Dod:Mal:Man75}, \cite%
{Dodonov:Man'koFIAN87}, \cite{Malkin:Man'ko79} and \cite{Malk:Man:Trif73})
can be presented as follows%
\begin{eqnarray}
P_{1} &=&\left( \kappa _{1}p+\frac{\left( c+d\right) \kappa _{1}-\kappa
_{1}^{\prime }}{2a}x\right) \exp \left( \int_{0}^{t}\left( c-d\right) \
ds\right) , \\
P_{2} &=&\left( \kappa _{2}p+\frac{\left( c+d\right) \kappa _{2}-\kappa
_{2}^{\prime }}{2a}x\right) \exp \left( \int_{0}^{t}\left( c-d\right) \
ds\right)
\end{eqnarray}%
in terms of two linearly independent solutions, $\kappa _{1}$ and $\kappa
_{2},$ of the homogeneous equation (\ref{AuxEquation}) when $C_{0}=0.$
Comparing these expressions with the form of the quadratic invariant (\ref%
{InvSymmForm}) at $C_{0}=0$ (no \textquotedblleft
potential\textquotedblright , a \textquotedblleft free
particle\textquotedblright ), when the operator square is complete, one can
treat the linear invariants as \textquotedblleft operator square
roots\textquotedblright\ \cite{Deb:Mikus} of the special quadratic
invariants (see also Lemma~2 regarding a convenient common factor).

\section{Quadratic Invariants and Cauchy Initial Value Problem}

Our decomposition (\ref{Linear&Quadratic}) of the quantum quadratic
invariant in terms of products of the linear ones not only results in the
Pinney solution (\ref{AuxSol})--(\ref{AuxSolWronskian}) of the corresponding
generalized Ermakov system (\ref{AuxEquation}) in a form of an
\textquotedblleft operator law of cosines\textquotedblright , but also
provides a somewhat better understanding, with the help of Lemma~1 and
properties of the linear invariants discussed in Section~4, how the
quadratic invariants act on solutions of the original time-dependent Schr%
\"{o}dinger equation. Indeed by (\ref{KSolution}) for two different
solutions, say $A_{1}=\mu _{1}$ and $A_{2}=\mu _{2},$ of the characteristic
equation, (\ref{EqC}), we have in an operator form:%
\begin{equation}
\psi \left( x,t\right) =K_{1}\left( t\right) \left[ \chi _{1}\left( y\right) %
\right] =K_{2}\left( t\right) \left[ \chi _{2}\left( y\right) \right]
\label{K12Solution}
\end{equation}%
in view of uniqueness of the Cauchy initial value problem (see, for example,
Refs.~\cite{Cor-Sot:Sua:SusInv} and \cite{Suaz:Sus}). Then by (\ref%
{SecondKSolution}):%
\begin{equation}
P_{1}\psi =e^{\int_{0}^{t}\left( c-d\right) ds}\ K_{1}\left( y\chi
_{1}\right) ,\quad P_{2}\psi =e^{\int_{0}^{t}\left( c-d\right) ds}\
K_{2}\left( y\chi _{2}\right)  \label{P12Act}
\end{equation}%
and%
\begin{eqnarray}
E\psi &=&e^{-\int_{0}^{t}\left( c-d\right) ds}\ \left[ C_{1}P_{1}^{2}\psi
+C_{2}P_{2}^{2}\psi +C_{3}\left( P_{1}P_{2}+P_{2}P_{1}\right) \psi \right]
\label{EAct} \\
&=&e^{\int_{0}^{t}\left( c-d\right) ds}\ \left[ C_{1}K_{1}\left( y^{2}\chi
_{1}\right) +C_{2}K_{2}\left( y^{2}\chi _{2}\right) \right]  \notag \\
&&+C_{3}\left[ P_{1}\left( K_{2}\left( y\chi _{2}\right) \right)
+P_{2}\left( K_{1}\left( y\chi _{1}\right) \right) \right] .  \notag
\end{eqnarray}%
We have%
\begin{eqnarray*}
&&K_{2}\left( y\chi _{2}\right) =K_{1}\left( \varphi _{1}\right) ,\qquad
\varphi _{1}=K_{1}^{-1}\left( 0\right) \left[ K_{2}\left( 0\right) \left(
y\chi _{2}\right) \right] , \\
&&K_{1}\left( y\chi _{1}\right) =K_{2}\left( \varphi _{2}\right) ,\qquad
\varphi _{2}=K_{2}^{-1}\left( 0\right) \left[ K_{1}\left( 0\right) \left(
y\chi _{1}\right) \right]
\end{eqnarray*}%
with the help of an analog of the Fourier transform. Therefore%
\begin{equation*}
P_{1}\left( K_{2}\left( y\chi _{2}\right) \right) =e^{\int_{0}^{t}\left(
c-d\right) ds}\ K_{1}\left( y\varphi _{1}\right) ,\quad P_{2}\left(
K_{1}\left( y\chi _{1}\right) \right) =e^{\int_{0}^{t}\left( c-d\right) ds}\
K_{2}\left( y\varphi _{2}\right) ,
\end{equation*}%
and finally%
\begin{equation}
E\psi =e^{\int_{0}^{t}\left( c-d\right) ds}\ \left[ C_{1}K_{1}\left(
y^{2}\chi _{1}\right) +C_{2}K_{2}\left( y^{2}\chi _{2}\right) +C_{3}\left(
K_{1}\left( y\varphi _{1}\right) +K_{2}\left( y\varphi _{2}\right) \right) %
\right] ,  \label{ESolution}
\end{equation}%
where each term satisfies the corresponding Schr\"{o}dinger equation. By the
superposition principle we arrive at a new solution in a complete agreement
with our Lemma~1. The corresponding initial data follow from (\ref{ESolution}%
) at $t=0$ and the time-evolution operator (\ref{CauchyInVProb}) can be
applied (see also Eq.~(\ref{SpectralDecomposition}) for an eigenfunction
expansion).\medskip

On the second hand one can always expand a square-integrable solution (\ref%
{KSolution}) of the Cauchy initial value problem in the standard form%
\begin{eqnarray}
\psi \left( x,t\right) &=&\int_{-\infty }^{\infty }K\left( x,y,t\right) \
\chi \left( y\right) \ dy  \label{KEexpansion} \\
&=&\sum_{n=0}^{\infty }c_{n}\left( t\right) \ \Psi _{n}\left( x,t\right) , 
\notag
\end{eqnarray}%
where%
\begin{equation}
c_{n}\left( t\right) =\int_{-\infty }^{\infty }\Psi _{n}^{\ast }\left(
x,t\right) \psi \left( x,t\right) \ dx  \label{KEexpansionConst}
\end{equation}%
(we use the asterisk for complex conjugate) by the Riesz--Fisher theorem 
\cite{Reed:SimonI}, \cite{Reed:SimonII}, \cite{Rudin} due to completeness of
the eigenfunctions (\ref{Eeigenfs})--(\ref{Eeigenfsconstant}) at all times.
Then by the Fubini theorem:%
\begin{equation}
c_{n}\left( t\right) =\int_{-\infty }^{\infty }\chi \left( y\right) \left(
\int_{-\infty }^{\infty }\Psi _{n}^{\ast }\left( x,t\right) K\left(
x,y,t\right) \ \ dx\right) \ dy  \label{KEexpansionConstant}
\end{equation}%
and the second integral can be evaluated with the help of Eqs.~(\ref{KKernel}%
), (\ref{Eeigenfs})--(\ref{Eeigenfsconstant}) and (\ref{Erd}) as follows%
\begin{eqnarray}
&&\int_{-\infty }^{\infty }\Psi _{n}^{\ast }\left( x,t\right) K\left(
x,y,t\right) \ \ dx  \label{KPsiIntegral} \\
&&\ =\frac{i^{n}e^{-i\left( n+1/2\right) \varphi }}{\left( \sqrt{\pi }%
2^{n}n!\right) ^{1/2}}\left( \frac{\sqrt{C_{0}}}{C_{0}\left( \frac{\kappa
_{1}}{\kappa }\right) ^{2}+\left( \frac{\kappa _{1}\kappa ^{\prime }-\kappa
_{1}^{\prime }\kappa }{2a}\right) ^{2}}\right) ^{1/4}\exp \left( \frac{1}{2}%
\int_{0}^{t}\left( d-c\right) \ ds\right)  \notag \\
&&\quad \times \exp \left( i\left( \gamma +\frac{\beta ^{2}\left( 0\right)
\kappa _{1}^{2}\left( 0\right) }{4a\kappa _{1}}\frac{\kappa \left( \kappa
_{1}\kappa ^{\prime }-\kappa _{1}^{\prime }\kappa \right) }{C_{0}\left( 
\frac{\kappa _{1}}{\kappa }\right) ^{2}+\left( \frac{\kappa _{1}\kappa
^{\prime }-\kappa _{1}^{\prime }\kappa }{2a}\right) ^{2}}\right) y^{2}\right)
\notag \\
&&\quad \times \exp \left( -\frac{\beta ^{2}\left( 0\right) \kappa
_{1}^{2}\left( 0\right) \sqrt{C_{0}}}{C_{0}\left( \frac{\kappa _{1}}{\kappa }%
\right) ^{2}+\left( \frac{\kappa _{1}\kappa ^{\prime }-\kappa _{1}^{\prime
}\kappa }{2a}\right) ^{2}}\frac{y^{2}}{2}\right) H_{n}\left( \frac{\beta
\left( 0\right) \kappa _{1}\left( 0\right) C_{0}^{1/4}y}{\sqrt{C_{0}\left( 
\frac{\kappa _{1}}{\kappa }\right) ^{2}+\left( \frac{\kappa _{1}\kappa
^{\prime }-\kappa _{1}^{\prime }\kappa }{2a}\right) ^{2}}}\right) .  \notag
\end{eqnarray}%
Here $\kappa _{1}=\mu \exp \left( \int_{0}^{t}\left( d-c\right) \ ds\right) $
and $\kappa $ are the corresponding solutions of auxiliary equation (\ref%
{AuxEquation}) with $C_{0}=0$ and $C_{0}\neq 0,$ respectively, with the
Wronskian $W\left( \kappa _{1},\kappa \right) =\kappa _{1}\kappa ^{\prime
}-\kappa _{1}^{\prime }\kappa $ and 
\begin{equation}
\tan \varphi =\frac{\kappa }{\kappa _{1}}\frac{\kappa _{1}\kappa ^{\prime
}-\kappa _{1}^{\prime }\kappa }{2a\sqrt{C_{0}}},  \label{tangent}
\end{equation}%
\begin{equation}
C_{0}\left( \frac{\kappa _{1}}{\kappa }\right) ^{2}+\left( \frac{\kappa
_{1}\kappa ^{\prime }-\kappa _{1}^{\prime }\kappa }{2a}\right) ^{2}=constant
\label{constant}
\end{equation}%
by (\ref{A4}). (The computational details are left to the reader; our
equation (\ref{constant}) is equivalent to the classical Ermakov invariant 
\cite{Ermakov}.) We may choose $\beta \left( 0\right) \kappa _{1}\left(
0\right) =1$ and arrive at the following result.

\begin{theorem}
\textbf{(Eigenfunction Expansions)} Solution of the Cauchy initial value
problem (\ref{KSolution}) in $L^{2}\left( 
\mathbb{R}
\right) $ can be obtained as an infinite series of multiples of the
quadratic invariant eigenfunctions (\ref{Eeigenfs}):%
\begin{equation}
\psi \left( x,t\right) =\sum_{n=0}^{\infty }c_{n}\left( t\right) \ \Psi
_{n}\left( x,t\right) ,  \label{QIExpSolution}
\end{equation}%
where the time-dependent coefficients are given by%
\begin{eqnarray}
c_{n}\left( t\right) &=&i^{n}\left( \frac{\delta }{\sqrt{\pi }2^{n}n!}%
\right) ^{1/2}e^{-i\left( n+1/2\right) \varphi }\exp \left( \frac{1}{2}%
\int_{0}^{t}\left( d-c\right) \ ds\right)  \label{c(t)constant} \\
&&\times \int_{-\infty }^{\infty }\exp \left( i\xi y^{2}\right) e^{-\delta
^{2}y^{2}/2}H_{n}\left( \delta y\right) \chi \left( y\right) \ dy.  \notag
\end{eqnarray}%
Here $\kappa _{1}\left( t\right) $and $\kappa \left( t\right) $ are
real-valued solutions of the homogeneous and nonhomogeneous auxiliary
equations (\ref{AuxEquation}), respectively, with the Wronskian $W\left(
t\right) =\kappa _{1}\kappa ^{\prime }-\kappa _{1}^{\prime }\kappa .$ The
phases $\varphi \left( t\right) $ and $\gamma \left( t\right) $ are
determined in terms of these solutions as follows%
\begin{equation}
\varphi =\arctan \left( \frac{\kappa }{\kappa _{1}}\frac{\kappa _{1}\kappa
^{\prime }-\kappa _{1}^{\prime }\kappa }{2a\sqrt{C_{0}}}\right) ,\quad \frac{%
d\varphi }{dt}=2\sqrt{C_{0}}\frac{a}{\kappa ^{2}},  \label{phase}
\end{equation}%
\begin{equation}
\frac{d\gamma }{dt}+\frac{a}{\kappa _{1}^{2}}=0  \label{gamma(t)}
\end{equation}%
and%
\begin{equation}
\delta =\frac{C_{0}^{1/4}}{\sqrt{C_{0}\left( \frac{\kappa _{1}}{\kappa }%
\right) ^{2}+\left( \frac{\kappa _{1}\kappa ^{\prime }-\kappa _{1}^{\prime
}\kappa }{2a}\right) ^{2}}}>0,  \label{deltaconstant}
\end{equation}%
\begin{equation}
\xi =\gamma +\frac{\delta ^{2}}{2\sqrt{C_{0}}}\frac{\left( \kappa _{1}\kappa
^{\prime }-\kappa _{1}^{\prime }\kappa \right) \kappa }{2a\kappa _{1}}
\label{xiconstant}
\end{equation}%
are constants. A spectral decomposition of the quadratic invariant $E$ in
the space of $L^{2}$-solutions is given by%
\begin{equation}
E\left( t\right) \psi \left( x,t\right) =\omega \left( t\right)
\sum_{n=0}^{\infty }c_{n}\left( t\right) \left( n+\frac{1}{2}\right) \ \Psi
_{n}\left( x,t\right) ,  \label{SpectralDecomposition}
\end{equation}%
where the \textquotedblleft frequency\textquotedblright\ $\omega \left(
t\right) $ is defined by (\ref{omega(t)}).
\end{theorem}

It is worth noting that%
\begin{equation}
\frac{d}{dt}\left[ \gamma +\frac{\delta ^{2}}{2\sqrt{C_{0}}}\frac{\kappa }{%
\kappa _{1}}\left( \frac{W}{2a}\right) \right] =0  \label{PhaseConstant}
\end{equation}%
by (\ref{gamma(t)}) and (\ref{A6}), which means that the phase factor (\ref%
{xiconstant}) in front of $y^{2}$ in the second integral of Eq.~(\ref%
{c(t)constant}) is indeed a constant. The second equation (\ref{phase})
follows from (\ref{PhaseConstant}) with the help of (\ref{tangent}) and (\ref%
{gamma(t)}).\medskip

Finally by choosing in Eqs.~(\ref{QIExpSolution})--(\ref{c(t)constant}) a
special (square) integrable initial data of the form%
\begin{equation}
\chi _{m}\left( y\right) =\exp \left( -i\xi y^{2}\right) e^{-\delta
^{2}y^{2}/2}H_{m}\left( \delta y\right)  \label{chispecial}
\end{equation}%
we conclude in view of the orthogonality property of Hermite polynomials
that the time-dependent wave functions%
\begin{eqnarray}
\psi _{m}\left( x,t\right) &=&D_{m}\exp \left( \frac{1}{2}\int_{0}^{t}\left(
d-c\right) \ ds\right)  \label{psispecial} \\
&&\times e^{-i\left( m+1/2\right) \varphi \left( t\right) }\ \Psi _{m}\left(
x,t\right)  \notag
\end{eqnarray}%
are particular solutions of the Schr\"{o}dinger equation (\ref{Schroedinger}%
)--(\ref{GenHam}) for arbitrary constants $D_{m}:$%
\begin{equation}
i\frac{\partial \psi _{m}}{\partial t}=H\psi _{m},\qquad E\psi _{m}=\omega
\left( m+\frac{1}{2}\right) \psi _{m}.  \label{Schroedingereigenfunctions}
\end{equation}%
They are also eigenfunctions of the quadratic invariant $E$\ corresponding
to the discrete \textquotedblleft spectrum\textquotedblright :%
\begin{equation}
\left\langle E\right\rangle _{m}=\int_{-\infty }^{\infty }\psi _{m}^{\ast
}E\psi _{m}\ dx=\left\vert D_{m}\right\vert ^{2}\omega _{0}\left( m+\frac{1}{%
2}\right) ,  \label{psispectrum}
\end{equation}%
see (\ref{Eeigenvp}). (It can be verified by direct substitution; the
details are left to the reader.) The explicit wave functions (\ref%
{psispecial}) are derived here without separation of the variables with the
aid of our Theorem~2 and certain variants of the Ermakov invariant. If $%
\left\vert D_{m}\right\vert =1,$ solution (\ref{QIExpSolution}) takes the
form%
\begin{equation}
\psi \left( x,t\right) =\sum_{n=0}^{\infty }i^{n}e^{-i\left( n+1/2\right)
\varphi \left( 0\right) }\ \psi _{n}\left( x,t\right) \int_{-\infty
}^{\infty }\psi _{n}^{\ast }\left( y,0\right) \chi \left( y\right) \ dy
\label{wavefunctionexpan}
\end{equation}%
in terms of the time-dependent wave functions (\ref{psispecial}) provided
that%
\begin{equation}
\varepsilon \left( 0\right) =\frac{C_{0}^{1/4}}{\kappa \left( 0\right) }%
=\delta ,\quad \frac{\left( c\left( 0\right) +d\left( 0\right) \right)
\kappa \left( 0\right) -\kappa ^{\prime }\left( 0\right) }{4a\left( 0\right)
\kappa \left( 0\right) }=\xi  \label{inconsts}
\end{equation}%
(see also \cite{Yeon:Lee:Um:George:Pandey93} for the path-integral method).
The traditional operator approach (for the parametric oscillator) is
presented in \cite{Lewis:Riesen69} and/or elsewhere (see also \cite%
{Andrio:Leach05} and \cite{Leach90}).

\section{A General Nonlinear Superposition Principle for Ermakov's Equations}

The Pinney superposition formula (\ref{AuxSol})--(\ref{AuxSolWronskian})
allows one to construct solutions of the nonlinear auxiliary equation (\ref%
{AuxEquation}) in terms of given solutions of the corresponding linear
equation. In general we take two different solutions, say $\kappa _{1}$ and $%
\kappa _{2},$ of the generalized Ermakov equations (\ref{AuxEquation}) with $%
C_{0}^{\left( 1,2\right) }\neq 0$ and consider two quadratic invariants:%
\begin{eqnarray}
E_{1}\left( t\right) &=&\left[ \left( \kappa _{1}\ p+\frac{\left( c+d\right)
\kappa _{1}-\kappa _{1}^{\prime }}{2a}\ x\right) ^{2}+\frac{C_{0}^{\left(
1\right) }}{\kappa _{1}^{2}}\ x^{2}\right] \lambda \left( t\right) , \\
E_{2}\left( t\right) &=&\left[ \left( \kappa _{2}\ p+\frac{\left( c+d\right)
\kappa _{2}-\kappa _{2}^{\prime }}{2a}\ x\right) ^{2}+\frac{C_{0}^{\left(
2\right) }}{\kappa _{2}^{2}}\ x^{2}\right] \lambda \left( t\right) .
\end{eqnarray}%
Their arbitrary linear combination,%
\begin{equation}
D_{1}E_{1}\left( t\right) +D_{2}E_{2}\left( t\right) =E\left( t\right)
\end{equation}%
($D_{1}$ and $D_{2}$ are constants), is also a quadratic invariant given by (%
\ref{InvSymmForm}) for a certain solution $\kappa $ of the nonlinear
auxiliary equation (\ref{AuxEquation}). Thus the following operator identity
holds%
\begin{eqnarray}
&&\left( \kappa \ p+\frac{\left( c+d\right) \kappa -\kappa ^{\prime }}{2a}\
x\right) ^{2}+\frac{C_{0}}{\kappa ^{2}}\ x^{2}
\label{QuadraticInvariantIdentity} \\
&&\qquad =D_{1}\left[ \left( \kappa _{1}\ p+\frac{\left( c+d\right) \kappa
_{1}-\kappa _{1}^{\prime }}{2a}\ x\right) ^{2}+\frac{C_{0}^{\left( 1\right) }%
}{\kappa _{1}^{2}}\ x^{2}\right]  \notag \\
&&\qquad \quad +D_{2}\left[ \left( \kappa _{2}\ p+\frac{\left( c+d\right)
\kappa _{2}-\kappa _{2}^{\prime }}{2a}\ x\right) ^{2}+\frac{C_{0}^{\left(
2\right) }}{\kappa _{2}^{2}}\ x^{2}\right]  \notag
\end{eqnarray}%
and in a similar fashion we arrive at a general nonlinear superposition
principle for the solutions of Ermakov's equations (\ref{AuxEquation}):%
\begin{equation}
\kappa ^{2}\left( t\right) =D_{1}\kappa _{1}^{2}\left( t\right) +D_{2}\kappa
_{2}^{2}\left( t\right) ,  \label{GenNonLinSupPrinciple}
\end{equation}%
where the constant $C_{0}$ is given by%
\begin{equation}
C_{0}-C_{0}^{\left( 1\right) }D_{1}^{2}-C_{0}^{\left( 2\right)
}D_{2}^{2}=D_{1}D_{2}\left[ \frac{W^{2}\left( \kappa _{1},\kappa _{2}\right) 
}{\left( 2a\right) ^{2}}+C_{0}^{\left( 1\right) }\frac{\kappa _{2}^{2}}{%
\kappa _{1}^{2}}+C_{0}^{\left( 2\right) }\frac{\kappa _{1}^{2}}{\kappa
_{2}^{2}}\right] \text{\label{Constant}}
\end{equation}%
with $W\left( \kappa _{1},\kappa _{2}\right) $ being the Wronskian of two
solutions (see Appendix~A). One can also derive this property by adding the
corresponding solutions (\ref{SolA})--(\ref{SolC}) of the original linear
system (\ref{EqA})--(\ref{EqCOne}) or with the help of the Pinney formula (%
\ref{AuxSol})--(\ref{AuxSolWronskian}). The details are left to the
reader.\medskip

\noindent \textbf{Acknowledgments.\/} The author thanks George E. Andrews,
Carlos Castillo-Ch\'{a}vez, Victor V. Dodonov, Vladimir~I. Man'ko, Svetlana
Roudenko and Kurt Bernardo Wolf for support, valuable comments and
encouragement. I am grateful to Peter~G.~L.~Leach for careful reading of the
manuscript --- his numerous suggestions have helped to improve the
presentation.

\appendix

\section{Required Transformations}

Equating coefficients in front of the operators $p^{2},$ $x^{2}$ and $px+xp$
in (\ref{OpIdentity})--(\ref{OpIdentityAB}) with the help of (\ref%
{MuKappaSubstitution}) one gets%
\begin{equation}
\kappa ^{2}=C_{1}\kappa _{1}^{2}+C_{2}\kappa _{2}^{2}+2C_{3}\kappa
_{1}\kappa _{2},  \label{A1}
\end{equation}%
\begin{eqnarray}
&&\left[ \frac{\left( c+d\right) \kappa -\kappa ^{\prime }}{2a}\right] ^{2}+%
\frac{C_{0}}{\kappa ^{2}}  \label{A2} \\
&&\quad =C_{1}\left[ \frac{\left( c+d\right) \kappa _{1}-\kappa _{1}^{\prime
}}{2a}\right] ^{2}+C_{2}\left[ \frac{\left( c+d\right) \kappa _{2}-\kappa
_{2}^{\prime }}{2a}\right] ^{2}  \notag \\
&&\qquad +2C_{3}\frac{\left[ \left( c+d\right) \kappa _{1}-\kappa
_{1}^{\prime }\right] \left[ \left( c+d\right) \kappa _{2}-\kappa
_{2}^{\prime }\right] }{\left( 2a\right) ^{2}},  \notag
\end{eqnarray}%
\begin{eqnarray}
&&\kappa \frac{\left( c+d\right) \kappa -\kappa ^{\prime }}{2a}  \label{A3}
\\
&&\quad =C_{1}\kappa _{1}\frac{\left( c+d\right) \kappa _{1}-\kappa
_{1}^{\prime }}{2a}+C_{2}\kappa _{2}\frac{\left( c+d\right) \kappa
_{2}-\kappa _{2}^{\prime }}{2a}  \notag \\
&&\qquad +C_{3}\left[ \kappa _{1}\frac{\left( c+d\right) \kappa _{2}-\kappa
_{2}^{\prime }}{2a}+\kappa _{2}\frac{\left( c+d\right) \kappa _{1}-\kappa
_{1}^{\prime }}{2a}\right] ,  \notag
\end{eqnarray}%
respectively. Multiply (\ref{A2}) by (\ref{A1}) and use (\ref{A3}) in the
left hand side in order to obtain (\ref{AuxSolWronskian}) as a result of
elementary but rather tedious calculations.\medskip

A general nonlinear superposition principle for Ermakov's equations, given
by (\ref{GenNonLinSupPrinciple})--(\ref{Constant}), is derived from the
quadratic invariant identity (\ref{QuadraticInvariantIdentity}) in a similar
fashion. In addition the constant in the right hand side of (\ref{Constant})
(see also (\ref{constant})) can be verified by direct differentiation as
follows%
\begin{eqnarray}
&&\frac{d}{dt}\left[ \left( \frac{W}{2a}\right) ^{2}+C_{0}^{\left( 1\right)
}\left( \frac{\kappa _{2}}{\kappa _{1}}\right) ^{2}+C_{0}^{\left( 2\right)
}\left( \frac{\kappa _{1}}{\kappa _{2}}\right) ^{2}\right]  \label{A4} \\
&&\quad =\frac{2W}{\left( 2a\right) ^{2}}\left( W^{\prime }-\frac{a^{\prime }%
}{a}W\right) +2W\left( C_{0}^{\left( 1\right) }\frac{\kappa _{2}}{\kappa
_{1}^{3}}-C_{0}^{\left( 2\right) }\frac{\kappa _{1}}{\kappa _{2}^{3}}\right)
\notag \\
&&\quad =\frac{2W}{\left( 2a\right) ^{2}}\left( W^{\prime }-\frac{a^{\prime }%
}{a}W\right)  \notag \\
&&\qquad -\frac{2W}{\left( 2a\right) ^{2}}\left[ \left( \kappa _{1}\kappa
_{2}^{\prime }-\kappa _{1}^{\prime }\kappa _{2}\right) ^{\prime }-\frac{%
a^{\prime }}{a}\left( \kappa _{1}\kappa _{2}^{\prime }-\kappa _{1}^{\prime
}\kappa _{2}\right) \right] =0  \notag
\end{eqnarray}%
with the help of auxiliary equations (\ref{AuxEquation}). (It is a natural
generalization of the classical Ermakov invariant \cite{Ermakov}.) En route,
we derive an identity:%
\begin{equation}
\frac{1}{2a}\left( \frac{W}{2a}\right) ^{\prime }+C_{0}^{\left( 1\right) }%
\frac{\kappa _{2}}{\kappa _{1}^{3}}-C_{0}^{\left( 2\right) }\frac{\kappa _{1}%
}{\kappa _{2}^{3}}=0,  \label{A5}
\end{equation}%
which may be considered as an extension of the familiar Abel theorem,
occurring when $C_{0}^{\left( 1\right) }=C_{0}^{\left( 2\right) }=0,$ to the
nonlinear Ermakov equations. Another identity,%
\begin{eqnarray}
&&\frac{d}{dt}\left[ \frac{\kappa _{2}}{\kappa _{1}}\left( \frac{W}{2a}%
\right) \right]  \label{A6} \\
&&\quad =\frac{2a}{\kappa _{1}^{2}}\left[ \left( \frac{W}{2a}\right)
^{2}-C_{0}^{\left( 1\right) }\left( \frac{\kappa _{2}}{\kappa _{1}}\right)
^{2}+C_{0}^{\left( 2\right) }\left( \frac{\kappa _{1}}{\kappa _{2}}\right)
^{2}\right] ,  \notag
\end{eqnarray}%
has been useful in Section~7 with $C_{0}^{\left( 1\right) }=0,$ $%
C_{0}^{\left( 2\right) }=C_{0}$ and $\kappa _{2}=\kappa $ when the right
hand side reduces to a multiple of the Ermakov invariant (see Eq.~(\ref%
{PhaseConstant})).

\section{An Integral Evaluation}

In Section~7 (see Eq.~(\ref{KPsiIntegral})) we use the following integral%
\begin{eqnarray}
&&\int_{-\infty }^{\infty }e^{-\lambda ^{2}\left( x-y\right)
^{2}}H_{n}\left( ay\right) \ dy  \label{Erd} \\
&&\quad =\frac{\sqrt{\pi }}{\lambda ^{n+1}}\left( \lambda ^{2}-a^{2}\right)
^{n/2}H_{n}\left( \frac{\lambda ax}{\left( \lambda ^{2}-a^{2}\right) ^{1/2}}%
\right) ,\quad \func{Re}\lambda ^{2}>0,  \notag
\end{eqnarray}%
which is equivalent to Eq.~(30) on page 195 of Vol.~2 of Ref.~\cite{Erd}
(the Gauss transform of Hermite polynomials), or Eq.~(17) on page 290 of
Vol.~2 of Ref.~\cite{ErdInt}.

\end{document}